\documentclass[a4paper,UKenglish]{lipics} 

\usepackage{microtype}

\usepackage{cite}
\usepackage{amsmath}
\usepackage{algorithmicx}
\usepackage{algpseudocode}
\usepackage{graphicx}

\title {Lempel-Ziv Factorization May Be Harder Than Computing All Runs}
\author{Dmitry Kosolobov}
\affil{Ural Federal University\\
  Ekaterinburg, Russia\\
  \texttt{dkosolobov@mail.ru}}

\authorrunning{D. Kosolobov}
\Copyright{Dmitry Kosolobov}


\algtext*{EndFunction}
\algtext*{EndProcedure}
\algtext*{EndFor}
\algtext*{EndWhile}
\algtext*{EndIf}
\algloopdefx{NIf}[1]{\textbf{if} #1 \textbf{then}}
\algloopdefx{NElse}{\textbf{else}}
\algloopdefx{NElseIf}{\textbf{else if}}
\algloopdefx{NForAll}[1]{\textbf{for each} #1 \textbf{do}}
\algloopdefx{NWhile}[1]{\textbf{while} #1 \textbf{do}}
\algloopdefx{NFor}[1]{\textbf{for} #1 \textbf{do}}

\serieslogo{}
\volumeinfo
  {Billy Editor and Bill Editors}
  {2}
  {Conference title on which this volume is based on}
  {1}
  {1}
  {1}
\EventShortName{}
\DOI{10.4230/LIPIcs.xxx.yyy.p}

\begin{document}

\maketitle

\begin{abstract}
The complexity of computing the Lempel-Ziv factorization and the set of all runs (= maximal repetitions) is studied in the decision tree model of computation over ordered alphabet. It is known that both these problems can be solved by RAM algorithms in $O(n\log\sigma)$ time, where $n$ is the length of the input string and $\sigma$ is the number of distinct letters in it. We prove an $\Omega(n\log\sigma)$ lower bound on the number of comparisons required to construct the Lempel-Ziv factorization and thereby conclude that a popular technique of computation of runs using the Lempel-Ziv factorization cannot achieve an $o(n\log\sigma)$ time bound. In contrast with this, we exhibit an $O(n)$ decision tree algorithm finding all runs in a string. Therefore, in the decision tree model the runs problem is easier than the Lempel-Ziv factorization. Thus we support the conjecture that there is a linear RAM algorithm finding all runs.
\end{abstract}

\section{Introduction}

String repetitions called runs and the Lempel-Ziv factorization are structures that are of a great importance for data compression and play a significant role in stringology. Recall that a run of a string is a nonextendable (with the same minimal period) substring whose minimal period is at most half of its length. The definition of the Lempel-Ziv factorization is given below. In the decision tree model, a widely used model to obtain lower bounds on the time complexity of various algorithms, we consider algorithms finding these structures. We prove that any algorithm finding the Lempel-Ziv factorization on a general ordered alphabet must perform $\Omega(n\log\sigma)$\footnote{Throughout the paper, $\log$ denotes the logarithm with the base~$2$.} comparisons in the worst case, where $n$ denotes the length of input string and $\sigma$ denotes the number of distinct letters in it. Since until recently, the only known efficient way to find all runs of a string was to use the Lempel-Ziv factorization, one might expect that there is a nontrivial lower bound in the decision tree model on the number of comparisons in algorithms finding all runs. These expectations were also supported by the existence of such a bound in the case of unordered alphabet. In this paper we obtain a somewhat surprising fact: in the decision tree model with an ordered alphabet, there exists a linear algorithm finding all runs. This can be interpreted as one cannot have lower bounds on the decision tree model for algorithms finding runs (a similar result for another problem is provided in~\cite{AhoHirschbergUllman} for example) but on the other hand, this result supports the conjecture by Breslauer~\cite[Chapter~4]{Breslauer} that there is a linear RAM algorithm finding all runs.

The Lempel-Ziv factorization \cite{LempelZiv} is a basic technique for data compression and plays an important role in stringology. It has several modifications used in various compression schemes. The factorization considered in this paper is used in LZ77-based compression methods. All known efficient algorithms for computation of the Lempel-Ziv factorization on a general ordered alphabet work in $O(n\log\sigma)$ time (see~\cite{Crochemore,EvenPrattRodeh,FialaGreene}), though all these algorithms are time and space consuming in practice. However for the case of polynomially bounded integer alphabet, there are efficient linear algorithms~\cite{AbouelhodaKurtzOhlenbusch,ChenPuglisiSmyth,CrochemoreIlieSmyth} and space efficient online algorithms~\cite{OkanoharaSadakane,Starikovskaya,YamamotoIBannaiEtal}.

Repetitions of strings are fundamental objects in both stringology and combinatorics on words. The notion of run, introduced by Main in~\cite{Main}, allows to grasp the whole periodic structure of a given string in a relatively simple form. In the case of unordered alphabet, there are some limitations on the efficiency of algorithms finding periodicities; in particular, it is known \cite{MainLorentz} that any algorithm that decides whether an input string over a general unordered alphabet has at least one run, requires $\Omega(n\log n)$ comparisons in the worst case. In~\cite{KolpakovKucherov}, Kolpakov and Kucherov proved that any string of length $n$ contains $O(n)$ runs and proposed a RAM algorithm finding all runs in linear time provided the Lempel-Ziv factorization is given. Thereafter much work has been done on the analysis of runs (e.g. see~\cite{CrochemoreIlieTinta, CKRRW, KolpakovPodolskiyPosypkinKhrapov}) but until the recent paper~\cite{BannaiIInenagaNakashimaTakedaTsuruta}, all efficient algorithms finding all runs of a string on a general ordered alphabet used the Lempel-Ziv factorization as a basis. Bannai et al.~\cite{BannaiIInenagaNakashimaTakedaTsuruta} use a different method based on Lyndon factorization but unfortunately, their algorithm spends $O(n\log\sigma)$ time too. Clearly, due to the found lower bound, our linear algorithm finding all runs doesn't use the Lempel-Ziv factorization yet our approach differs from that of~\cite{BannaiIInenagaNakashimaTakedaTsuruta}.

The paper is organized as follows. Section~\ref{SectPrel} contains some basic definitions used throughout the paper. In Section~\ref{SectLempelZiv} we give a lower bound on the number of comparisons required to construct the Lempel-Ziv factorization. In Section~\ref{SectRuns} we present additional definitions and combinatorial facts that are necessary for Section~\ref{SectLin}, where we describe our linear decision tree algorithm finding all runs.

\section{Preliminaries}\label{SectPrel}

A \emph{string of length $n$} over the alphabet $\Sigma$ is a map $\{1,2,\ldots,n\} \mapsto \Sigma$, where $n$ is referred to as the length of $w$, denoted by $|w|$. We write $w[i]$ for the $i$th letter of $w$ and $w[i..j]$ for $w[i]w[i{+}1]\ldots w[j]$. Let $w[i..j]$ be the empty string for any~$i > j$. A string $u$ is a \emph{substring} (or a \emph{factor}) of $w$ if $u=w[i..j]$ for some $i$ and $j$. The pair $(i,j)$ is not necessarily unique; we say that $i$ specifies an \emph{occurrence} of $u$ in $w$. A string can have many occurrences in another string. An integer $p$ is a \emph{period} of $w$ if $0 < p < |w|$ and $w[i] = w[i{+}p]$ for $i=1,\ldots,|w|{-}p$. For any integers $i,j$, the set $\{k\in \mathbb{Z} \colon i \le k \le j\}$ (possibly empty) is denoted by $\overline{i, j}$.

The only computational model that is used in this paper is the \emph{decision tree} model. Informally, a decision tree processes input strings of given \emph{fixed} length and each path starting at the root of the tree represents the sequence of pairwise comparisons made between various letters in the string. The computation follows an appropriate path from the root to a leaf; each leaf represents a particular answer to the studied problem.

More formally, a decision tree processing strings of length $n$ is a rooted directed ternary tree in which each interior vertex is labeled with an ordered pair $(i,j)$ of integers, $1\le i,j\le n$, and edges are labeled with the symbols ``$<$'', ``$=$'', ``$>$'' (see Fig.~\ref{fig:tree}). The \emph{height} of a decision tree is the number of edges in the longest path from the root to a leaf of the tree. Consider a path $p$ connecting the root of a fixed decision tree to some vertex $v$. Let $t$ be a string of length $n$. Suppose that $p$ satisfies the following condition: it contains a vertex labeled with a pair $(i,j)$ with the outgoing edge labeled with $<$ (resp., $>$, $=$) if and only if $t[i]<t[j]$ (resp., $t[i]>t[j]$,  $t[i]=t[j]$). Then we say that the vertex $v$ is \emph{reachable} by the string $t$ or the string $t$ \emph{reaches} the vertex $v$. Clearly, each string reaches exactly one leaf of any given tree.
\begin{figure}[htb]
\centering
\includegraphics[scale=0.55]{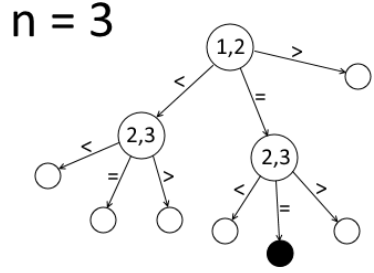}
\caption{A decision tree of height $2$ processing strings of length $3$. The strings $aaa$ and $bbb$ reach the shaded vertex.}
\label{fig:tree}
\end{figure}

\section{A Lower Bound on Algorithms Computing the Lempel-Ziv Factorization} \label{SectLempelZiv}

The Lempel-Ziv factorization of a string $t$ is the decomposition $t = t_1t_2\cdots t_k$, built by the following greedy procedure processing $t$ from left to right:
\begin{itemize}
\item $t_1 = t[1]$;
\item let $t_1\cdots t_{i-1} = t[1..j]$; if $t[j{+}1]$ does not occur in $t[1..j]$, put $t_i=t[j{+}1]$; otherwise, put $t_i$ to be the longest prefix of $t[j{+}1..n]$ that has an occurrence starting at some position $\le j$.
\end{itemize}
For example, the string $abababaabbbaaba$ has the Lempel-Ziv factorization $a.b.ababa.ab.bb.aab.a$.

Let $t$ and $t'$ be strings of length $n$. Suppose $t = t_1t_2\ldots t_k$ and $t' = t'_1t'_2\ldots t'_{k'}$ are their Lempel-Ziv factorizations. We say that the Lempel-Ziv factorizations of $t$ and $t'$ are equivalent if $k = k'$ and $|t_i| = |t'_i|$ for each $i\in \overline{1,k}$. We say that a decision tree processing strings of length $n$ finds the Lempel-Ziv factorization if for any strings $t$ and $t'$ of length $n$ such that $t$ and $t'$ reach the same leaf of the tree, the Lempel-Ziv factorizations of $t$ and $t'$ are equivalent.

\begin{theorem}
The construction of the Lempel-Ziv factorization for a string of length $n$ with at most $\sigma$ distinct letters requires $\Omega(n\log\sigma)$ comparisons of letters in the worst case. \label{LempelZiv}
\end{theorem}
\begin{proof}
Let $a_1 < \ldots < a_\sigma$ be an alphabet. To obtain the lower bound, we construct a set of input strings of length $n$ such that the construction of the Lempel-Ziv factorization for these strings requires performing $\Theta(n)$ binary searches on the $\Theta(\sigma)$-element alphabet.

Without loss of generality, we assume that $n$ and $\sigma$ are even and $2 < \sigma < n/2$. Denote $s_1 = a_1a_3a_5\ldots a_{\sigma{-}1}$, $s_2 = a_\sigma a_2 a_\sigma a_4 \ldots a_\sigma a_{\sigma{-}2}a_\sigma a_\sigma$, and $s = s_1s_2$. We view $s$ as a ``dictionary'' containing all letters $a_i$ with even $i$. Note that $|s| = 1.5\sigma$. Consider a string $t$ of the following form:
\begin{equation}
\begin{array}{l}
a_\sigma a_{i_1} a_\sigma a_{i_2} \ldots a_\sigma a_{i_k}a_\sigma a_\sigma,\\
\mbox{where }k = \frac{n - 1.5\sigma - 2}2\mbox{ and }i_j \in \overline{2,\sigma{-}2}\mbox{ is even for any }j\in \overline{1,k}\enspace. \label{eq:form}
\end{array}
\end{equation}
Informally, the string $t$ represents a sequence of queries to our ``dictionary'' $s$; any decision tree finding the Lempel-Ziv factorization of the string $st$ must identify each $a_{i_j}$ of $t$ with some letter of $s$. Otherwise, we can replace $a_{i_j}$ with the letter $a_{i_j-1}$ or $a_{i_j+1}$ thus changing the Lempel-Ziv factorization of the whole string; the details are provided below. Obviously, $|s| + |t| = n$ and there are $(\sigma/2 - 1)^k$ possible strings $t$ of the form~\eqref{eq:form}. Let us take a decision tree which computes the Lempel-Ziv factorization for the strings of length $n$. It suffices to prove that each leaf of this tree is reachable by at most one string $st$ with $t$ of the form~\eqref{eq:form}. Indeed, such decision tree has at least $(\sigma/2 - 1)^k$ leafs and the height of the tree is at least $\log_3 ((\sigma/2 - 1)^k) = k\log_3(\sigma/2 - 1) = \Omega(n\log\sigma)$.

Suppose to the contrary that some leaf of the decision tree is reachable by two distinct strings $r = st$ and $r' = st'$ such that $t$ and $t'$ are of the form~\eqref{eq:form}; then for some $l\in\overline{1,n}$, $r'[l] \ne r[l]$. Obviously $l = |s| + 2l'$ for some $l'\in \overline{1,k}$ and therefore $r[l] = a_p$ for some even $p \in \overline{2,\sigma{-}2}$. Suppose $r'[l] < r[l]$. Let $l_1 < \ldots < l_m$ be the set of all integers $l' > |s|$ such that for any string $t_0$ of the form~\eqref{eq:form}, if the string $r_0 = st_0$ reaches the same leaf as the string $r$, then $r_0[l'] = r_0[l]$. Consider a string $r''$ that differs from $r$ only in the letters $r''[l_1], \ldots, r''[l_m]$ and put $r''[l_1] = \ldots = r''[l_m] = a_{p-1}$. Let us first prove that the string $r''$ reaches the same leaf as $r$. Consider a vertex of the path connecting the root and the leaf reachable by $r$. Let the vertex be labeled with a pair $(i,j)$. We have to prove that the comparison of $r''[i]$ and $r''[j]$ leads to the same result as the comparison of $r[i]$ and $r[j]$. The following cases are possible:
\begin{enumerate}
\item $i, j \ne l_q$ for all $q\in \overline{1,m}$; then $r[i] = r''[i]$ and $r[j] = r''[j]$;
\item $i = l_q$ for some $q\in \overline{1,m}$ and $r[i] < r[j]$; then since $r''[l_q] = a_{p-1} < a_p = r[l_q] = r[i]$ and $r[j] = r''[j]$, we obtain $r''[i] < r''[j]$;
\item $i = l_q$ for some $q\in \overline{1,m}$ and $r[i] > r[j]$; then we have $j \ne p/2$ because $r[p/2] = r'[p/2] = a_{p-1} > r'[i]$ while $r'[i] > r'[j]$, and thus since $r[i] = a_p > r[j]$, we see that $a_{p-1} = r''[i] > r[j] = r''[j]$;
\item $i = l_q$ for some $q\in \overline{1,m}$ and $r[i] = r[j]$; then, by definition of the set $\{l_1,\ldots,l_m\}$, $j = l_{q'}$ for some $q'\in \overline{1,m}$ and $r''[i] = r''[j] = a_{p-1}$;
\item $j = l_q$ for some $q\in \overline{1,m}$; this case is symmetric to the above cases.
\end{enumerate}

Thus $r''$ reaches the same leaf as $r$. But the strings $r$ and $r''$ have the different Lempel-Ziv factorizations: the Lempel-Ziv factorization of $r''$ has one letter factor $a_{p-1}$ at position $l_1$ while $r$ does not since $r[l_1{-}1..l_1{+}1] = a_\sigma a_pa_\sigma$ is a substring of $s = r[1..|s|]$. This contradicts to the fact that the analyzed tree computes the Lempel-Ziv factorization.
\end{proof}

\section{Runs} \label{SectRuns}

In this section we consider some combinatorial facts that will be useful in our main algorithm described in the following section.

The \emph{exponent} of a string $t$ is the number $|t| / p$, where $p$ is the minimal period of $t$. A \emph{run} of a string $t$ is a substring $t[i..j]$ of exponent at least~$2$ and such that both substrings $t[i{-}1..j]$ and $t[i..j{+}1]$, if defined, have strictly greater minimal periods than $t[i..j]$. A run whose exponent is greater than or equal to~$3$ is called a \emph{cubic run}. For a fixed $d \ge 1$, a \emph{$d$-short run} of a string $t$ is a substring $t[i..j]$ which can be represented as $xyx$ for nonempty strings $x$ and $y$ such that $0 < |y| \le d$, $|x|$ is the minimal period of $t[i..j]$, and both substrings $t[i{-}1..j]$ and $t[i..j{+}1]$, if defined, have strictly greater minimal periods.
\begin{example}
The string $t = aabaabab$ has four runs $t[1..2] = aa$, $t[4..5] = aa$, $t[1..7] = aabaaba$, $t[5..8] = abab$ and one $1$-short run $t[2..4] = aba$. The sum of exponents of all runs is equal to $2 + 2 + \frac{7}{3} + 2 \approx 8.33$.
\end{example}

As it was proved in \cite{KolpakovKucherov}, the number of all runs is linear in the length of string. We use a stronger version of this fact.

\begin{lemma}[{see~\cite[Theorem 9]{BannaiIInenagaNakashimaTakedaTsuruta}}]
The number of all runs in any string of length $n$ is less than $n$. \label{RunsNum}
\end{lemma}


The following lemma is a straightforward corollary of \cite[Lemma 1]{KolpakovPodolskiyPosypkinKhrapov}.

\begin{lemma}[{see~\cite{KolpakovPodolskiyPosypkinKhrapov}}]
For a fixed $d \ge 1$, any string of length $n$ contains $O(n)$ $d$-short runs. \label{AlmostRunsNum}
\end{lemma}

We also need a classical property of periodic strings.

\begin{lemma}[see~\cite{FineWilf}]
Suppose a string $w$ has periods $p$ and $q$ such that $p + q - \gcd(p,q) \le |w|$; then $\gcd(p,q)$ is a period of $w$.
\label{FineWilfLemma}
\end{lemma}
\begin{lemma}
Let $t_1$ and $t_2$ be substrings with the periods $p_1$ and $p_2$ respectively. Suppose $t_1$ and $t_2$ have a common substring of the length $p_1 + p_2 - \gcd(p_1, p_2)$ or greater; then $t_1$ and $t_2$ have the period $\gcd(p_1, p_2)$. \label{StringsIntersect}
\end{lemma}
\begin{proof}
It is immediate from Lemma~\ref{FineWilfLemma}.
\end{proof}

Unfortunately, in a string of length $n$ the sum of exponents of runs with the minimal period $p$ or greater is not equal to $O(\frac{n}{p})$ as the following example from \cite{Kolpakov} shows: $(01)^k(10)^k$. Indeed, for any $p < 2k$, the string $(01)^k(10)^k$ contains at least $k - \lfloor p/2\rfloor$ runs with the shortest period $p$ or greater: $1(01)^i(10)^i1$ for $i \in \overline{\lfloor p/2\rfloor,k{-}1}$. However, it turns out that this property holds for cubic runs.

\begin{lemma}
For any $p\ge 2$ and any string $t$ of length $n$, the sum of exponents of all cubic runs in $t$ with the minimal period $p$ or greater is less than $\frac{12n}{p}$. \label{CubicRunExp}
\end{lemma}
\begin{proof}
Consider a string $t$ of length $n$. Denote by $\mathcal{R}$ the set of all cubic runs of $t$. Let $t_1 = t[i_1..j_1]$ and $t_2 = t[i_2..j_2]$ be distinct cubic runs such that $i_1 \le i_2$. For any string $u$, $e(u)$ denotes the exponent of $u$ and $p(u)$ denotes the minimal period of $u$. It follows from Lemma~\ref{StringsIntersect} that $t_1$ and $t_2$ cannot have a common substring of length $p(t_1) + p(t_2)$ or longer. Let $\delta$ be a positive integer. Suppose $2\delta \le p(t_1),p(t_2) \le 3\delta$; then either $j_1 < i_2$ or $j_1 - i_2 < p(t_1) + p(t_2) \le 2.5 p(t_1)$. The later easily implies $i_2 - i_1 > \delta$ and therefore $\rho = |\{u \in \mathcal{R}\colon 2\delta \le p(u) \le 3\delta\}| < \frac{n}{\delta}$. Moreover, we have $i_2 - i_1 \ge (e(t_1) - 2.5)p(t_1) \ge (e(t_1) - 2.5)2\delta$. Hence $\sum\limits_{u \in \mathcal{R}, 2\delta \le p(u) \le 3\delta} (e(u) - 2.5)2\delta \le n$ and then $\sum\limits_{u \in \mathcal{R}, 2\delta \le p(u) \le 3\delta} e(u) \le \frac{n}{2\delta} + 2.5\rho < \frac{3n}{\delta}$.

Denote $\delta_i = (\frac{3}{2})^i$ and $k = \lfloor \log_{\frac{3}{2}} \frac{p}{2}\rfloor$. Evidently $(\frac{2}{3})^k \ge \frac{4}{3p}$. Finally, we obtain $\sum\limits_{u \in \mathcal{R}, p(u) \ge p} e(u) < \sum_{i = k}^\infty \frac{3n}{\delta_i} = \sum_{i = k}^\infty 3n(\frac{2}{3})^i = 3n \frac{(2/3)^k}{1/3} \le 9n\frac{4}{3p}  = \frac{12n}{p}$.
\end{proof}

\section{Linear Decision Tree Algorithm Finding All Runs} \label{SectLin}

We say that a decision tree processing strings of length $n$ \emph{finds all runs with a given property $P$} if for each distinct strings $t_1$ and $t_2$ such that $|t_1| = |t_2| = n$ and $t_1$ and $t_2$ reach the same leaf of the tree, the substring $t_1[i..j]$ is a run satisfying $P$ iff $t_2[i..j]$ is a run satisfying $P$ for all $i, j \in \overline{1,n}$.

We say that two decision trees processing strings of length $n$ are equivalent if for each reachable leaf $a$ of the first tree, there is a leaf $b$ of the second tree such that for any string $t$ of length $n$, $t$ reaches $a$ iff $t$ reaches $b$. The \emph{basic height} of a decision tree is the minimal number $k$ such that each path connecting the root and a leaf of the tree has at most $k$ edges labeled with the symbols ``$<$'' and ``$>$''.

For a given positive integer $p$, we say that a run $r$ of a string is \emph{$p$-periodic} if $2p \le |r|$ and $p$ is a (not necessarily minimal) period of $r$. We say that a run is a \emph{$p$-run} if it is $q$-periodic for some $q$ which is a multiple of $p$. Note that any run is $1$-run.

\begin{example}
Let us describe a ``naive'' decision tree finding all $p$-runs in strings of length $n$. Denote by $t$ the input string. Our tree simply compares $t[i]$ and $t[j]$ for all $i, j \in \overline{1,n}$ such that $|i - j|$ is a multiple of $p$. The tree has the height $\sum_{i=1}^{\lfloor n/p\rfloor} (n - ip) = O(n^2/p)$ and the same basic height.
\end{example}

Note that a decision tree algorithm finding runs doesn't report runs in the same way as RAM algorithms do. The algorithm only collects sufficient information to conclude where the runs are; once its knowledge of the structure of the input string becomes sufficient to find all runs without further comparisons of symbols, the algorithm stops and doesn't care about the processing of obtained information. To simplify the construction of an efficient decision tree, we use the following lemma that enables us to estimate only the basic height of our tree.
\begin{lemma}
Suppose a decision tree processing strings of length $n$ has basic height~$k$. Then it is equivalent to a decision tree of height $\le k + n$. \label{EqualComp}
\end{lemma}
\begin{proof}
To construct the required decision tree of height $\le k+n$, we modify the given decision tree of basic height $k$. First, we remove all unreachable vertices of this tree. After this, we contract each non-branching path into a single edge, removing all intermediate vertices and their outgoing edges. Indeed, the result of a comparison corresponding to such an intermediate vertex is determined by the previous comparisons. So, it is straightforward that the result tree is equivalent to the original tree. Now it suffices to prove that there are at most $n{-}1$ edges labeled with the symbol ``$=$'' along any path connecting the root and some leaf.

Observe that if we perform $n{-}1$ comparisons on $n$ elements and each comparison yields an equality, then either all elements are equal or the result of at least one comparison can be deduced by transitivity from other comparisons. Suppose a path connecting the root and some leaf has at least $n$ edges labeled with the symbol ``$=$''. By the above observation, the path contains an edge labeled with ``$=$'' leaving a vertex labeled with $(i,j)$ such that the equality of the $i$th and the $j$th letters of the input string follows by transitivity from the comparisons made earlier along this path. Then this vertex has only one reachable child. But this is impossible because all such vertices of the original tree were removed during the contraction step. This contradiction finishes the proof.
\end{proof}

\begin{lemma}
For any integers $n$ and $p$, there is a decision tree that finds all $p$-periodic runs in strings of length $n$ and has basic height at most $2\lceil n/p\rceil$.\label{ConstRun}
\end{lemma}
\begin{proof}
Denote by $t$ the input string. The algorithm is as follows (note that the resulting decision tree contains only comparisons of letters of $t$):
\begin{enumerate}
\item assign $i \gets 1$;
\item if $t[i] \ne t[i{+}p]$, then assign $i \gets i + p$, $h \gets \min\{i, n - p\}$ and for $i' = h{-}1, h{-}2, \ldots$, compare $t[i']$ and $t[i'{+}p]$ until $t[i'] \ne t[i'{+}p]$; \label{lst:pnt}
\item increment $i$ and if $i \le n - p$, jump to line \ref{lst:pnt}.
\end{enumerate}
Obviously, the algorithm performs at most $2\lceil n/p\rceil$ symbol comparisons yielding inequalities. Let us prove that the algorithm finds all $p$-periodic runs.

Let $t[j..k]$ be a $p$-periodic run. For the sake of simplicity, suppose $1 < j < k < n$. To discover this run, one must compare $t[l]$ and $t[l{+}p]$ for each $l\in \overline{j{-}1,k{-}p{+}1}$. Let us show that the algorithm performs all these comparisons. Suppose, to the contrary, for some $l \in \overline{j{-}1,k{-}p{+}1}$, the algorithm doesn't compare $t[l]$ and $t[l{+}p]$. Then for some $i_0$ such that $i_0 < l < i_0 + p$, the algorithm detects that $t[i_0] \ne t[i_0{+}p]$ and ``jumps'' over $l$ by assigning $i = i_0 + p$ at line~\ref{lst:pnt}. Obviously $i_0 < j$. Then $h = \min\{i_0 + p, n - p\} < k$ and hence for each $i' = h{-}1, h{-}2, \ldots, j{-}1$, the algorithm compares $t[i']$ and $t[i'{+}p]$. Since $j - 1 \le l < i_0 + p$, $t[l]$ and $t[l{+}p]$ are compared, contradicting to our assumption.
\end{proof}

\newcommand{\sgn}{\operatorname{sgn}}
\begin{theorem}
There is a constant $c$ such that for any integer $n$, there exists a decision tree of height at most $cn$ that finds all runs in strings of length $n$.
\end{theorem}
\begin{proof}
By Lemma~\ref{EqualComp}, it is sufficient to build a decision tree with linear basic height. So, below we count only the comparisons yielding inequalities and refer to them as ``inequality comparisons''. In fact we prove the following more general fact: for a given string $t$ of length $n$ and a positive integer $p$, we find all $p$-runs performing $O(n/p)$ inequality comparisons. To find all runs of a string, we simply put $p = 1$.

The algorithm consists of five steps. Each step finds $p$-runs of $t$ with a given property. Let us choose a positive integer constant $d \ge 2$ (the exact value is defined below.) The algorithm is roughly as follows:
\begin{enumerate}
\item find in a straightforward manner all $p$-runs having periods $\le dp$;
\item using the information from step~1, build a new string $t'$ of length $n/p$ such that periodic factors of $t$ and $t'$ are strongly related to each other;
\item find $p$-runs of $t$ related to periodic factors of $t'$ with exponents less than~$3$;
\item find $p$-runs of $t$ related to periodic factors of $t'$ with periods less than~$d$;
\item find $p$-runs of $t$ related to other periodic factors of $t'$ by calling steps~1--5 recursively for some substrings of~$t$.
\end{enumerate}

\textbf{Step 1.} Initially, we split the string $t$ into $n/p$ contiguous blocks of length $p$ (if $n$ is not a multiple of $p$, we pad $t$ on the right to the required length with a special symbol which is less than all other symbols.) For each $i\in\overline{1,n/p}$ and $j \in \overline{1,d}$, we denote by $m_{i,j}$ the minimal $k\in \overline{1,p}$ such that $t[(i{-}1)p{+}k] \ne t[(i{-}1)p{+}k{+}jp]$ and we put $m_{i,j} = -1$ if $ip + jp > n$ or there is no such $k$. To compute $m_{i,j}$, we simply compare $t[(i{-}1)p{+}k]$ and $t[(i{-}1)p{+}k{+}jp]$ for $k = 1,2,\ldots, p$ until $t[(i{-}1)p{+}k] \ne t[(i{-}1)p{+}k{+}jp]$.

\begin{example} \label{texample}
Let $t = bbba\cdot aada\cdot aaaa\cdot aaaa\cdot aada\cdot aaaa\cdot aaab\cdot bbbb\cdot bbbb$, $p = 4$, $d = 2$. The following table contains $m_{i,j}$ for $j = 1,2$:
$$
\begin{array}{r||c|c|c|c|c|c|c|c|c}
i                   &  1      &  2      &  3      &  4      &  5      &  6      &  7      &  8       &   9  \\
\hline
\hline
t[(i{-}1)p{+}1..ip] &  bbba   &  aada   &  aaaa   &  aaaa   &  aada   &  aaaa   &  aaab   &   bbbb   &  bbbb \\
\hline
m_{i,1},m_{i,2}    &  1,1    &  3,3    &  -1,3    &  3,-1   &  3,3    &  4,1    &  1,1    &   -1,-1  &  -1,-1\\
\end{array}
$$
\end{example}
To compute a particular value of $m_{i,j}$, one needs at most one inequality comparison (zero inequality comparisons if the computed value is~$-1$.) Further, for each $i\in\overline{1,n/p}$ and $j \in \overline{1,d}$, we compare $t[ip{-}k]$ and $t[ip{-}k{+}jp]$ (if defined) for $k = 0, 1, \ldots, p{-}1$ until $t[ip{-}k] \ne t[ip{-}k{+}jp]$; similar to the above computation of $m_{i,j}$, this procedure performs at most one inequality comparison for any given $i$ and $j$. Hence, the total number of inequality comparisons is at most $2dn/p$. Once these comparisons are made, all $pq$-periodic runs in the input string are determined for all $q\in\overline{1,d}$.

\textbf{Step 2.} Now we build an auxiliary structure induced by $m_{i,j}$ on the string $t$. In this step, no comparisons are performed; we just establish some combinatorial properties required for further steps. We make use of the function:
$$
\sgn(a,b) = \left\{ \begin{array}{ll}-1,& a < b,\\ \phantom{-}0,& a = b,\\ \phantom{-}1, & a > b\enspace. \end{array} \right.
$$
We create a new string $t'$ of length $n/p$. The alphabet of this string can be taken arbitrary, we just describe which letters of $t'$ coincide and which do not. For each $i_1, i_2 \in \overline{1,n/p}$, $t'[i_1] = t'[i_2]$ iff for each $j \in \overline{1,d{-}1}$, either $m_{i_1,j} = m_{i_2,j} = -1$ or the following conditions hold simultaneously:
$$
\begin{array}{l}
m_{i_1,j} \ne -1, m_{i_2,j} \ne -1,\\
m_{i_1,j} = m_{i_2,j},\\
\sgn(t[(i_1{-}1)p{+}m_{i_1,j}], t[(i_1{-}1)p{+}m_{i_1,j}{+}jp]) = \sgn(t[(i_2{-}1)p{+}m_{i_2}], t[(i_2{-}1)p{+}m_{i_2,j}{+}jp])\enspace.
\end{array}
$$
Note that the status of each of these conditions is known from step~1. Also note that the values $m_{i,d}$ are not used in the definition of $t'$; we computed them only to find all $dp$-periodic $p$-runs.

\addtocounter{theorem}{-1}
\begin{example}[continued]
Denote $s_i = \sgn(t[(i{-}1)p{+}m_{i,1}], t[(i{-}1)p{+}m_{i,1}{+}p])$. Let $\{e,f,g,h,i,j\}$ be a new alphabet for the string $t'$. The following table contains $m_{i,1}$, $s_i$, and $t'$:
$$
\begin{array}{r||c|c|c|c|c|c|c|c|c}
i                    &  1      &  2      &  3      &  4      &  5      &  6      &  7      &  8      &  9\\
\hline
\hline
t[(i{-}1)p{+}1..ip]  &  bbba   &  aada   &  aaaa   &  aaaa   &  aada   &  aaaa   &  aaab   &  bbbb   &  bbbb   \\
\hline
m_{i,1}              &    1    &    3    &    -1   &    3    &    3    &    4    &    1    &   -1    &   -1    \\
\hline
s_i                  &    1    &    1    &    -    &   -1    &    1    &   -1    &   -1    &   -     &    -    \\
\hline
t'[i]                &    j    &    e    &    f    &    g    &    e    &    h    &    i    &   f     &     f
\end{array}
$$
\end{example}

If $t$ contains two identical sequences of $d$ blocks each, i.e., $t[(i_1{-}1)p{+}1..(i_1{-}1{+}d)p] = t[(i_2{-}1)p{+}1..(i_2{-}1{+}d)p]$ for some $i_1$, $i_2$, then $m_{i_1,j} = m_{i_2,j}$ for each $j \in \overline{1,d{-}1}$ and hence $t'[i_1] = t'[i_2]$. This is why $t'[2] = t'[5]$ in Example~\ref{texample}. On the other hand, equal letters in $t'$ may correspond to different sequences of blocks in $t$, like the letters $t'[3] = t'[8]$ in Example~\ref{texample}. The latter property makes the subsequent argument more involved but allows us to keep the number of inequality comparisons linear. Let us point out the relations between periodic factors of $t$ and $t'$.

Let for some $q > d$, $t[k{+}1..k{+}l]$ be a $pq$-periodic $p$-run, i.e., $t[k{+}1..k{+}l]$ is a $p$-run that is not found on step~1. Denote $k' = \lceil k/p\rceil$. Since $t[k{+}1..k{+}l]$ is $pq$-periodic, $t'$ has some periodicity in the corresponding substring, namely, $u = t'[k'{+}1..k'{+}\lfloor l/p\rfloor{-}d]$ has the period $q$ (see example below). Let $t'[k_1..k_2]$ be the largest substring of $t'$ containing $u$ and having the period $q$. Since $2q \le \lfloor l/p\rfloor = |u| + d$, $t'[k_1..k_2]$ is either a $d$-short run with the minimal period $q$ or a run whose minimal period divides $q$.

\addtocounter{theorem}{-1}
\begin{example}[continued]
Consider Fig.~\ref{fig:corrun}. Let $k = 3$, $l = 24$. The string $t[k{+}1..k{+}l] = a\cdot aada\cdot aaaa\cdot aaaa\cdot aada\cdot aaaa\cdot aaa$ is a $p$-run with the minimal period $pq = 12$ (here $q = 3 > 2 = d$). Denote $k' = \lceil k/p \rceil = 1$, $k_1 = 2$, and $k_2 = 5$. The string $t'[k'{+}1..k'{+}\lfloor l/p\rfloor{-}d] = t'[k_1..k_2] = t'[2..5] = efge$ is a $d$-short run of $t'$ with the minimal period~$q = 3$.
\end{example}
\begin{figure}[htb]
\centering
\includegraphics[scale=0.55]{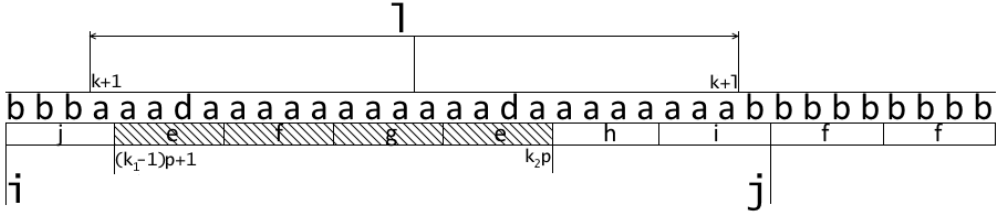}
\caption{A $p$-run corresponding to $d$-short run $t'[k_1..k_2] = efge$, where $k_1 = 2$, $k_2 = 5$, $p = 4$, $d = 2$, $q = 3$, $k = 3$, $l = 2pq = 24$, $i = (k_1{-}2)p{+}1 = 1$, $j = (k_2{+}d)p = 28$.}
\label{fig:corrun}
\end{figure}

Conversely, given a run or $d$-short run $t'[k_1..k_2]$ with the minimal period $q$, we say that a $p$-run $t[k{+}1..k{+}l]$ \emph{corresponds to $t'[k_1..k_2]$} (or $t[k{+}1..k{+}1]$ is a $p$-run \emph{corresponding to $t'[k_1..k_2]$}) if $t[k{+}1..k{+}l]$ is, for some integer $r$, $rpq$-periodic and $t'[k'{+}1..k'{+}\lfloor l/p\rfloor{-}d]$, where $k' = \lceil k/p\rceil$, is a substring of $t'[k_1..k_2]$ (see Fig.~\ref{fig:corrun} and Example~\ref{texample}). 

The above observation shows that each $p$-run of $t$ that is not found on step~1 corresponds to some run or $d$-short run of $t'$. Let us describe the substring that must contain all $p$-runs of $t$ corresponding to a given run or $d$-short run $t'[k_1..k_2]$. Denote $i = (k_1 - 2)p + 1$ and $j = (k_2 + d)p$. Now it is easy to see that if $t[k{+}1..k{+}l]$ is a $p$-run corresponding to $t'[k_1..k_2]$, then $t[k{+}1..k{+}l]$ is a substring of $t[i..j]$.

\addtocounter{theorem}{-1}
\begin{example}[continued]
For $k = 3$ and $l = 24$, the string $t[k{+}1..k{+}l] = a\cdot aada\cdot aaaa\cdot aaaa\cdot aada\cdot aaaa\cdot aaa$ is a $p$-run corresponding to $t'[k_1..k_2] = efge$, where $k_1 = 2$, $k_2 = 5$. Indeed, the string $t'[k'{+}1..k'{+}\lfloor l/p\rfloor{-}d] = t'[2..5]$, for $k' = \lceil k/p\rceil = 1$, is a substring of $t'[k_1..k_2]$. Denote $i = (k_1 - 2)p + 1 = 1$, $j = (k_2 + d)p = 28$. Observe that $t[k{+}1..k{+}l] = t[4..27]$ is a substring of $t[i..j] = t[1..28]$.
\end{example}

It is possible that there is another $p$-run of $t$ corresponding to the string $t'[k_1..k_2]$. Consider the following example.

\begin{example} \label{twoCorr}
Let $t = fabcdedabcdedaaifjfaaifjff$, $p = 2$, $d = 2$. Denote $s_i = \sgn(t[(i{-}1)p{+}m_{i,1}], t[(i{-}1)p{+}m_{i,1}{+}p])$. Let $\{w,x,y,z\}$ be a new alphabet for the string $t'$. The following table contains $m_{i,1}$, $s_i$, and $t'$:
$$
\begin{array}{r||c|c|c|c|c|c|c|c|c|c|c|c|c}
i                   & 1  & 2  & 3  & 4  & 5  & 6  & 7  & 8  & 9  & 10 & 11 & 12 & 13 \\
\hline
\hline
t[(i{-}1)p{+}1..ip] & fa & bc & de & da & bc & de & da & ai & fj & fa & ai & fj & ff \\
m_{i,1}             & 1  & 1  & 2  & 1  & 1  & 2  & 1  & 1  & 2  & 1  & 1  & 2  & -1 \\
s_i                 & 1  &-1  & 1  & 1  &-1  & 1  & 1  &-1  & 1  & 1  &-1  & 1  & -  \\
t'[i]               & x  & y  & z  & x  & y  & z  & x  & y  & z  & x  & y  & z  & w
\end{array}
$$
Note that $p$-runs $t[2..13] = abcded\cdot abcded$ and $t[14..25] = aaifjf\cdot aaifjf$ correspond to the same $p$-run of $t'$, namely, $t'[1..12] = xyz\cdot xyz\cdot xyz\cdot xyz$.
\end{example}

Thus to find for all $q > d$ all $pq$-periodic $p$-runs of $t$, we must process all runs and $d$-short runs of~$t'$.

\textbf{Step 3.} Consider a noncubic run $t'[k_1..k_2]$. Let $q$ be its minimal period. Denote $i = (k_1 - 2)p + 1$ and $j = (k_2 + d)p$. The above analysis shows that any $p$-run of $t$ corresponding to $t'[k_1..k_2]$ is a $p'$-periodic run of $t[i..j]$ for some $p' = pq, 2pq, \ldots, lpq$, where $l = \lfloor (j - i + 1) / (2pq) \rfloor$. Since $(k_2 - k_1 + 1)/q < 3$, we have $l = \lfloor(k_2 - k_1 + 2)/(2q) + d/(2q)\rfloor = O(d)$. Hence to find all $p$-runs of $t[i..j]$, it suffices to find for each $p' = pq, 2pq, \ldots, lpq$ all $p'$-periodic runs of $t[i..j]$ using Lemma~\ref{ConstRun}. Thus the processing performs $O(l(j - i + 1)/pq) = O(d^2) = O(1)$ inequality comparisons. Analogously we process $d$-short runs of $t'$. Therefore, by Lemmas~\ref{RunsNum} and~\ref{AlmostRunsNum}, only $O(|t'|) = O(n/p)$ inequality comparisons are required to process all $d$-short runs and noncubic runs of $t'$.

Now it suffices to find all $p$-runs of $t$ corresponding to cubic runs of $t'$.

\textbf{Step 4.} Let $t'[k_1..k_2]$ be a cubic run with the minimal period $q$. In this step we consider the case $q < d$. It turns out that such small-periodic substrings of $t'$ correspond to substrings in $t$ that are either periodic and discovered at step~1, or aperiodic. Therefore this step does not include any comparisons. The precise explanation follows.

Suppose that $m_{k,q} = -1$ for all $k\in \overline{k_1,k_1{+}q{-}1}$. Then $m_{k,q} = -1$ for all $k = k_1,\ldots ,k_2$ by periodicity of $t'[k_1..k_2]$. Therefore by the definition of $m_{k,q}$, we have $t[k] = t[k{+}pq]$ for all $k \in \overline{(k_1{-}1)p{+}1,k_2p}$. Hence the substring $t[(k_1{-}1)p{+}1..k_2p{+}pq]$ has the period $pq$. Now it follows from Lemma~\ref{StringsIntersect} that any $p$-run of $t$ corresponding to $t'[k_1..k_2]$ is $pq$-periodic and therefore was found on step~1 because $pq < dp$.

Suppose that $m_{k,q}\ne -1$ for some $k\in \overline{k_1,k_1{+}q{-}1}$. Denote $s = (k - 1)p + m_{k,q}$, $l = \lfloor (k_2p - s)/pq \rfloor + 1$. Let $r \in \overline{1,l}$. Since $t'[k] = t'[k{+}rq]$, we have $m_{k,q} = m_{k{+}rq,q}$ and $\sgn(t[s], t[s{+}pq]) = \sgn(t[s{+}rpq], t[s{+}(r{+}1)pq])$ (see Fig.~\ref{fig:cubicrun}). Therefore, one of the following sequences of inequalities holds:
\begin{equation}
\begin{array}{l}
t[s] < t[s{+}pq] < t[s{+}2pq] < \ldots < t[s{+}lpq], \\
t[s] > t[s{+}pq] > t[s{+}2pq] > \ldots > t[s{+}lpq]\enspace. \label{eq:chains}
\end{array}
\end{equation}
\begin{figure}[htb]
\centering
\includegraphics[scale=0.65]{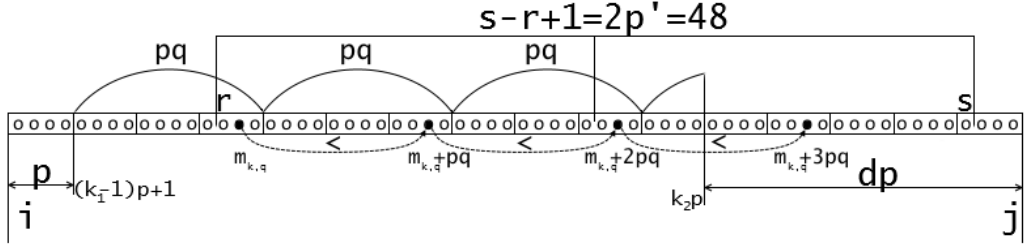}
\caption{A cubic run of $t'$ with the shortest period $q = 3 < d = 5$, where $p = 4$, $k_1 = 2$, $k_2 = 11$, $k = 4$, $m_{k,q} = 15$, $l = 3$, $p' = 2pq = 24$.}
\label{fig:cubicrun}
\end{figure}

Let $p'$ be a multiple of $pq$ such that $p' > dp$. Now it suffices to show that due to the found ``aperiodic chain'', there are no $p'$-periodic $p$-runs of $t$ corresponding to $t'[k_1..k_2]$.

Suppose, to the contrary, $t[r..s]$ is a $p'$-periodic $p$-run corresponding to $t'[k_1..k_2]$ (see Fig.~\ref{fig:cubicrun}). Denote $r' = \lceil (r - 1)/p\rceil$ and $l' = \lfloor (s - r + 1)/p\rfloor$. By the definition of corresponding $p$-runs, $u = t'[r'{+}1..r'{+}l'{-}d]$ is a substring of $t'[k_1..k_2]$. Since $s - r + 1 \ge 2p'$ and $p' > dp$, we have $|u| = l' - d \ge 2p'/p - d > p'/p$. Therefore, $r \le r'p + m_{r'{+}1,q} < r'p + m_{r'{+}1,q} + p' \le s$ and the inequalities~\eqref{eq:chains} imply $t[r'p{+}m_{r'{+}1,q}] \ne t[r'p{+}m_{r'{+}1,q} + p']$, a contradiction.

\textbf{Step 5.} Let $t'[k_1..k_2]$ be a cubic run with the minimal period $q$ such that $q \ge d$. Denote $i = (k_1 - 2)p + 1$ and $j = (k_2 + d)p$. To find all $p$-runs corresponding to the run $t'[k_1..k_2]$, we make a recursive call executing steps~1--5 again with new parameters $n = j - i + 1$, $p = pq$, and $t = t[i..j]$.

After the analysis of all cubic runs of $t'$, all $p$-runs of $t$ are found and the algorithm stops. Now it suffices to estimate the number of inequality comparisons performed during any run of the described algorithm.

\textbf{Time analysis.} As shown above, steps~1--4 require $O(n/p)$ inequality comparisons. Let $t'[i_1..j_1], \ldots, t'[i_k..j_k]$ be the set of all cubic runs of $t'$ with the minimal period $d$ or greater. For $l \in \overline{1,k}$, denote by $q_l$ the minimal period of $t'[i_l..j_l]$ and denote $n_l = j_l - i_l + 1$. Let $T(n, p)$ be the number of inequality comparisons required by the algorithm to find all $p$-runs in a string of length $n$. Then $T(n, p)$ can be computed by the following formula:
$$
T(n,p) = O\left(n/p\right) + T\left((n_1 + d + 1)p, pq_1\right) + \ldots + T\left((n_k + d + 1)p, pq_k\right)\enspace.
$$
For $l\in \overline{1,k}$, the number $n_l/q_l$ is, by definition, the exponent of $t'[i_l..j_l]$. It follows from Lemma~\ref{CubicRunExp} that the sum of exponents of all cubic runs of $t'$ with the shortest period $d$ or larger is less than $\frac{12n}{d}$. Note that for any $l\in \overline{1,k}$, $n_l \ge 3q_l \ge 3d$ and therefore $n_l + d + 1 < 2n_l$. Thus assuming $d = 48$, we obtain $\frac{(n_1 + d + 1)p}{pq_1} + \ldots + \frac{(n_k + d + 1)p}{pq_k} < \frac{2n_1}{q_1} + \ldots + \frac{2n_k}{q_k} \le \frac{24n}{dp} = \frac{n}{2p}$. Finally, we have $T(n, p) = O(\frac{n}{2^0p} + \frac{n}{2^1p} + \frac{n}{2^2p} + \ldots) = O(n/p)$. The reference to Lemma~\ref{EqualComp} ends the proof.
\end{proof}

\section{Conclusion} \label{SectConc}

It remains an open problem whether there exists a linear RAM algorithm finding all runs in a string over a general ordered alphabet. However, it is still possible that there are nontrivial lower bounds in some more sophisticated models that are strongly related to RAM model.

\subparagraph*{Acknowledgement}

The author would like to thank Arseny M. Shur for many valuable comments and the help in the preparation of this paper.

\bibliographystyle{splncs}

\end{document}